%% file: main.tex
\documentclass[conference,a4paper]{IEEEtran}
\input{packages}
\input{macros}

\title{Bounds and Codes for General Phased Burst Errors}
\author{%
  \IEEEauthorblockN{\textbf{Sebastian Bitzer}$^1$, \textbf{Andrea Di Giusto}$^2$, \textbf{Alberto Ravagnani}$^2$,  and \textbf{Eitan Yaakobi}$^3$}

  \IEEEauthorblockA{$^1$Technical University of Munich, Germany \\
                      $^2$Eindhoven University of Technology, the Netherlands \\   
                      $^3$Technion -- Israel Institute of Technology, Israel\\
                    \textit{sebastian.bitzer@tum.de}, \textit{\{a.di.giusto, a.ravagnani\}@tue.nl}, \textit{yaakobi@cs.technion.ac.il}
                    }
}

\newtheorem{question}{Question}
\newtheorem{property}{Property}

\begin{document}

\maketitle

\begin{abstract}
\input{sections/0_abstract}
\end{abstract}

\section{Introduction}\label{sect:intro}
\input{sections/1_intro_general}

\section{The PBE Adversarial Channel}\label{sec:prelim}
\input{sections/2_prelim_general}

\section{Bounds on the Maximal Code Size}\label{sec:bounds}
\input{sections/3_bounds_general}

\section{Code Construction}\label{sec:constr}
\input{sections/4_construction_general}

\section*{Acknowledgment}
The authors thank Hugo Sauerbier Couvée for \Cref{rem:both_relevant}.
Sebastian Bitzer  acknowledges the financial support by the Federal Ministry of Education and
Research of Germany in the program of “Souverän. Digital. Vernetzt.”. Joint project 6G-life, project
identification number: 16KISK002.
Andrea Di Giusto is supported by the European Commission through grant 101072316.

\bibliographystyle{plain}
\bibliography{references.bib}

\appendix
\input{sections/6_appendix_general}

\end{document}

%% file: packages.tex
\usepackage{amssymb, amsthm,bm,mathtools} 
\usepackage{enumitem}
\usepackage{color}
\usepackage{comment}
\usepackage[utf8]{inputenc} 
\usepackage[T1]{fontenc}
\usepackage{graphicx} 

\usepackage{mleftright} 
\usepackage[nameinlink,capitalize]{cleveref}
\usepackage[dvipsnames]{xcolor}
\usepackage{tikz}
\usepackage{pgfplots}
\usepackage{subcaption}
\usetikzlibrary{patterns}
\usepackage{url} 

\newtheorem{theorem}{Theorem}

\newtheorem{lemma}[theorem]{Lemma}

\newtheorem{remark}[theorem]{Remark}
\newtheorem{example}[theorem]{Example}
\newtheorem{corollary}[theorem]{Corollary}
\newtheorem{definition}[theorem]{Definition}
\newtheorem{construction}[theorem]{Construction}
\newtheorem{comparison}[theorem]{Comparison}

\Crefname{equation}{Eq.}{Eqs.}
\Crefname{figure}{Fig.}{Figs.}
\Crefname{construction}{Constr.}{Constrs.}
\Crefname{example}{Ex.}{Exs.}

\usepackage{algorithm}
\usepackage[noend]{algpseudocode}

%% file: macros.tex
\newcommand{\6}{\mathbf}

\newcommand{\wtH}{\textnormal{wt}_\textnormal{H}} 
\newcommand{\eset}{\mathcal{E}}
\newcommand{\deset}{\Delta(\mathcal{E})}
\newcommand{\Heset}{\textnormal{H}\text{-}\eset} 
\newcommand{\code}{\mathcal{C}}
\newcommand{\N}{\mathbb{N}}
\newcommand{\Z}{\mathbb{Z}}

\newcommand{\F}{\mathbb{F}}
\newcommand{\Fam}{\mathcal{F}}
\newcommand{\V}{\mathcal{V}}

\newcommand{\E}{\mathcal{E}}

\newcommand{\HC}{\textnormal{HC}}
\newcommand{\PBC}{\textup{PBC}}
\newcommand{\HPBC}{\textnormal{H}\text{-}\PBC}
\newcommand{\col}{\textnormal{col}}

\newcommand{\codeA}{\mathcal{A}}
\newcommand{\codeB}{\mathcal{B}}
\newcommand{\codeC}{\mathcal{C}}

\newcommand{\T}{T}
\newcommand{\W}{W}

\newcommand{\wtt}{t}
\newcommand{\wtw}{w}

\newcommand{\eseta}{\mathcal{E}_1}
\newcommand{\esetb}{\mathcal{E}_2}
\newcommand{\wtta}{t_1}
\newcommand{\wttb}{t_2}

\newcommand{\esetaa}{\mathcal{E}_{11}}
\newcommand{\esetbb}{\mathcal{E}_{22}}
\newcommand{\esetba}{\mathcal{E}_{21}}
\newcommand{\esetab}{\mathcal{E}_{12}}

\newcommand{\ca}{c_1}
\newcommand{\cb}{c_2}

\newcommand{\caa}{c_{11}}
\newcommand{\cbb}{c_{22}}
\newcommand{\cba}{c_{21}}
\newcommand{\cab}{c_{12}}

\newcommand{\HBall}[3]{B_{#1}(#2,#3)} 

\newcommand{\RGV}{R_\mathrm{GV}}
\newcommand{\RH}{R_\mathrm{H}}

\newcommand{\RThree}{R_\mathrm{3lvl}}
\newcommand{\RTwo}{R_\mathrm{2lvl}}

%% file: sections/0_abstract.tex
THIS PAPER IS ELIGIBLE FOR THE
STUDENT PAPER AWARD. 
Phased Burst Errors (PBEs) are bursts of errors occurring at one or more known locations.
The correction of PBEs is a classical topic in coding theory, with prominent applications such as the design of array codes for memory systems or distributed storage.
We propose a general yet fine-grained approach to this problem, accounting not only for the number of bursts but also the error structure in each burst.
By modeling PBEs as an error set in an adversarial channel, we investigate bounds on the maximal size of codes that can correct them.
The PBE-correction capability of generalized concatenated codes is analyzed, and asymptotically good PBE-correcting codes are constructed, recovering a classical construction in a specific problem instance.

%% file: sections/1_intro_general.tex
A classic topic in coding theory is the correction of bursts of errors with a known start location and maximal duration, which arise in various contexts, such as non-volatile memory systems \cite{dolecek2017channel}.
In mathematical terms, given an array $\6X\in\F_q^{n\times m}$, the goal is to correct errors restricted to an unknown subset of the columns of $\6X$.
Such error patterns, known as Phased Burst Errors (PBEs), were first studied in \cite{goodman1993phased} (single PBE) and \cite{blaum1993new} (multiple PBEs), and inspired major applications \cite{blaum1995evenodd}.

\textbf{This work.} %
We consider PBEs as an error set (PBE set) in the context of adversarial channels \cite{ravagnani2018adversarial}.
Let $\eseta\subseteq\esetb$ be arbitrary subsets of $\F_q^n$, and let $0\leq w\leq m$.
We say an array $\6X\in\F_q^{n\times m}$ is an $(\eseta,\esetb,w)$-PBE if it has at most $w$ columns in $\esetb$, and the remaining columns belong $\eseta$.
A code correcting all such patterns is an $(\eseta,\esetb,w)$-PBE correcting code, or $(\eseta,\esetb,w)$-PBECC for short; we address the following 
\begin{question}\label{quest:codesize}
    Given $n$, $m$, $\eseta\subseteq\esetb$, $w$, what is the maximum size $A_q(n,m,\eseta,\esetb,w)$ of an $(\eseta,\esetb,w)$-PBECC in $\F_q^{n\times m}$?
\end{question}
More precisely, we examine the asymptotic behavior for~$m$ growing linearly in~$n$ and $|\eseta|, |\esetb|$ scaling exponentially with $n$.
Analyzing these error bursts as an arbitrary error set \cite{loeliger1994basic}, we derive upper and lower bounds on the size of a maximal code correcting them.
As an application of our results, we will consider the case where $\eseta=\{\60\}$ and $\esetb$ is a Hamming ball centered around zero in $\F_q^n$, see \Cref{fig:channel_model}.
In this framework, we compare our new bounds with the classical Hamming and Gilbert-Varshamov (GV) bounds, demonstrating that the structure of the PBEs enables an asymptotic improvement in code rate for the same overall error weight.

The structure of generalized concatenated codes aligns naturally with the structure of phased burst errors, making them a suitable candidate for code constructions.
We provide general guarantees on the PBE-correction capability, which gives rise to explicit constructions.
Our analysis demonstrates that the constructed codes are asymptotically good, offering positive rates for all channel parameters with positive GV bound.
For Hamming-metric bursts, we recover the classical construction of \cite{wolf1965codes} and improve it in particular cases.

\textbf{Related work.} %
Although the interest in PBEs traces back to \cite{blaum1993new,goodman1993phased,roth1998reduced} (and \cite{blaum1995evenodd} as a prominent application), many recent works have revisited this topic,  sometimes under different terminologies.
Most of the literature concerns PBEs in the Hamming metric: in \cite{gabrys2012graded,roth2014coding}, phased bursts are considered, where the error columns can have either high or low Hamming weight.
More recently, LDPC codes for correcting phased bursts of erasures were studied in \cite{xiao2019quasi,li2021balanced}.
The performance of a PBEC code under random bursts of errors is investigated in \cite{raphaeli2005burst}.
For cryptographic applications of PBEs and generic error sets, see \cite{da1993secret,manganiello2023generic}.

\textbf{Structure.} %
\Cref{sec:prelim} introduces the necessary terminology on adversarial channels and arbitrary error sets and formally states the studied problem. In \Cref{sec:bounds}, we derive bounds on the one-shot capacity of the channels introduced in \Cref{sec:prelim}, using the approach of \cite{loeliger1994basic} to the study of arbitrary error sets. 
In \Cref{sec:constr}, we investigate the PBE-correction capability of generalized concatenated codes.
Based on this analysis, general PBEC codes are constructed, and their performance is compared with the derived bounds.

\begin{figure}[t]
    \centering
    \scalebox{0.95}{
    \input{figures/channel_model}
    }
    \caption{Two possible Hamming PBEs with $n = 5$, $t=2$, $m=6$, and $w=2$. That is, $\eseta = \{\60\}\subset\esetb=\HBall{2}{5}{2}$.}
    \label{fig:channel_model}
\end{figure}

\textbf{Artifacts.} %
Code to reproduce the figures and examples is available at \url{https://github.com/sebastianbitzer/pbe}. 
An extended version is available at \url{https://arxiv.org/abs/2501.12280}.

%% file: figures/channel_model.tex







\begin{tikzpicture}
\begin{scope}
\node [rectangle, fill = CornflowerBlue!40!white,minimum width = 0.5cm, minimum height = 2.5cm, fill opacity = 0.5] at (2.75,-1.25) (err1){};

\node [rectangle, fill = CornflowerBlue!40!white,minimum width = 0.5cm, minimum height = 2.5cm, fill opacity = 0.5] at (0.75,-1.25) (err2){};

\draw[] (0,0) grid[xstep=0.5, ystep=0.5] (3,-2.5);

\node[draw=black, thick, name = cw, rectangle, minimum width = 3cm, minimum height = 2.5cm] at (1.5,-1.25) {};

\node[anchor = south, align = center] at (cw.north) {$m$ blocks\\ $\leq w$ erroneous};
\node[anchor = south, rotate=90, align = center] at (cw.west) {$n$ symbols\\ $\leq t$ errors};

\node[] at (0.75,-0.25) {$1$};
\node[] at (0.75,-1.25) {$1$};
\node[] at (2.75,-2.25) {$1$};
\node[] at (2.75,-1.75) {$1$};
\end{scope}

\begin{scope}[shift={(3.5,0)}]

\node [rectangle, fill = CornflowerBlue!40!white,minimum width = 0.5cm, minimum height = 2.5cm, fill opacity = 0.5] at (2.25,-1.25) (err1){};

\node [rectangle, fill = CornflowerBlue!40!white,minimum width = 0.5cm, minimum height = 2.5cm, fill opacity = 0.5] at (0.75,-1.25) (err2){};

\draw[] (0,0) grid[xstep=0.5, ystep=0.5] (3,-2.5);

\node[draw=black, thick, name = cw, rectangle, minimum width = 3cm, minimum height = 2.5cm] at (1.5,-1.25) {};

\node[] at (0.75,-0.75) {$1$};
\node[] at (0.75,-1.75) {$1$};
\node[] at (2.25,-1.25) {$1$};

\end{scope}

\end{tikzpicture}

%% file: sections/2_prelim_general.tex
\subsection{Notation}
 
For a set $A$, let $|A|$ denote its cardinality and $2^A$ its power set.
Let $\V$ be a vector space, $A,B\subseteq\V$,
Then, the difference set of $A$ and $B$ is the set $\Delta(A,B)=\{a-b:a\in A, b\in B\}$. 
$\Delta(A)=\Delta(A,A)$ is the difference set of $A$.
For an integer $s$, $[s] = \{1,\ldots,s\}$.

For $q$ a prime power and an array $\6X\in\F_q^{n\times m}$,  let $\col(\6X)$ denote the multiset of its columns.
For $\6x\in\F_q^n$, we denote by $\wtH(\6x)$ its Hamming weight.
For $0\leq t\leq n$, we denote the Hamming-metric ball as $\HBall{q}{n}{t} = \{\6x\in\F_q^n: \wtH(\6x) \leq t\}$.
For $\T = \wtt/n$, it holds that
\begin{equation}\label{eq:general_approximations}
    q^{F_q(\T)n-o(n)}\leq|\HBall{q}{n}{t}|\leq q^{F_q(\T)n}\;, 
\end{equation}%
where $F_q(T) = H_q(\max\{T, \tfrac{q-1}{q}\})$ with $H_q(x)$ the $q$-ary entropy function.

\subsection{Problem Statement} 
 
\noindent We consider the following general communication framework.

\begin{definition}\label{def:adversarial_channel}
    Let $\V$ be a vector space over the finite field $\F_q$.
    An additive adversarial channel on $\V$ with error set $\eset\subseteq\V$ is a function $\Omega:\V\mapsto 2^\V$ associating to each $\6X\in\V$ its fan-out set $\Omega(\6X)=\{\6X+\6E: \6E\in\eset\}= \6X + \eset$. 
\end{definition}

All mentioned channels will respect this definition, and we will simply say that $\Omega$ is a channel over $\V$ throughout the paper.
This and some of the following definitions also work in more general frameworks, see \cite{ravagnani2018adversarial}.
A code in $\V$ is a subset $\code\subseteq\V$; a linear code $\code$ is a linear subspace of $\V$.
A one-shot code for $\Omega$ is a code $\code$ s.t. for any two distinct $X,X'\in\code$ we have $\Omega(\6X)\cap\Omega(\6X')=\emptyset$.
For a linear code, this is equivalent to $\codeC \cap \deset = \{\60\}$, see \cite{loeliger1994basic}.

\begin{example}\label{example:Hamming_channel_1}
    Let $\wtt\leq n\in\N$; a ubiquitous example of an additive adversarial channel is the Hamming Channel $\HC(n,t)$, where $\V=\F_q^n$ and $\eset=\{\6X\in\F_q^n:\;\wtH(\6X)\leq t\} = \HBall{q}{n}{t}$ is the Hamming ball of radius $t$.
    One-shot codes for this channel are precisely the $t$-error correcting codes in $\F_q^n$.
\end{example}

The \emph{one-shot capacity} of a channel $\Omega$ is defined as
\begin{equation}\label{eq:one_shot_capacity}
    C_1(\Omega)\coloneqq\max\biggl\{\!\frac{\log_q(|\code|)}{\dim(\V)}\!:\code\text{ is a one-shot code for }\Omega\!\biggr\}.\!
\end{equation}
$C_1(\Omega)$ is the normalized maximum amount of information that can be transmitted using the channel once and with zero error probability.
This notion is distinct from the \emph{zero error} capacity of a channel $C_0(\Omega)$, which is the maximum rate of error-free communication using the channel multiple times \cite{shannon1956zero,ravagnani2018adversarial}.

Bounds on the one-shot capacity of a channel $\Omega$ are derived by closely examining its error set \cite{loeliger1994basic}.
In particular, for a code $\code\subseteq\V$ with $|\code||\eset|>|\V|$, the pigeonhole principle implies that there exist distinct elements $\6X,\6X'\in\code$ such that $\Omega(\6X)\cap\Omega(\6X')\neq\emptyset$.
Then, $\code$ cannot be a one-shot code for $\Omega$, yielding the following upper bound on $C_1(\Omega)$.

\begin{lemma}[Proposition 1, \cite{loeliger1994basic}]\label{lem:Hamming_adversarial}
    For any channel $\Omega$ over $\V$,  
    \begin{equation*}
        C_1(\Omega)\leq1-\frac{\log_q(|\eset|)}{\dim(\V)}\;.
    \end{equation*}
\end{lemma}

A lower bound for the one-shot capacity can be established via an existence result for linear codes.
This is achieved by considering families of codes with random-like behavior.

\begin{definition}
    A nonempty set $\Fam$ of linear codes is \textit{balanced} if every $\6X\in\V\setminus\{\60\}$ belongs an equal number of codes in $\Fam$.
\end{definition}

Balancedness is the combinatorial equivalence of randomness, in the sense that a random linear code contains any nonzero element with the same probability, just like a code sampled uniformly at random from a balanced family.
Notice that the family of all linear codes in $\V$ of a given dimension, $\Fam_k=\{\code\subseteq\V\;:\;\dim(\code)=k\}$, is balanced.
The following result gives a sufficient condition for the existence of a linear one-shot code for a channel $\Omega$.

\begin{theorem}[Theorem 1, \cite{loeliger1994basic}]\label{thm:GV_generic}
    Let $\Omega$ be a channel over $\V$ with error set $\eset$ containing the zero vector, and let $\Fam$ be a balanced family of codes in $\V$. 
    If $|\Delta(\E)|\leq q^{n-k}$ then $\Fam$ contains a one-shot code for $\Omega$.
    It follows that
    \begin{equation*}
        C_1(\Omega)\geq1-\frac{\log_q(|\Delta(\eset)|)}{\dim(\V)}\;.
    \end{equation*}
\end{theorem} 

\begin{example}\label{example:Hamming_channel_2}
    Applying \Cref{lem:Hamming_adversarial} and \Cref{thm:GV_generic} to the channel $\HC(n,t)$ of \Cref{example:Hamming_channel_1}, we get the classical Hamming and Gilbert-Varshamov asymptotic bounds for the rate of block codes respectively.
    In fact, when $\eset=\HBall{q}{n}{t}$, the difference set is simply $\Delta(\eset)=\HBall{q}{n}{\min(n,2t)}$
    The asymptotic versions of the bounds are recovered by applying \Cref{eq:general_approximations}: we get    
    \begin{equation}
        1-F_q(t/n)\leq C_1(\HC(n,t))\leq1-F_q(t/n)+o(1).
    \end{equation}
\end{example}

Given the link to the classical counterparts, we will refer to \Cref{lem:Hamming_adversarial} and \Cref{thm:GV_generic} as generic \emph{Hamming} and \emph{Gilbert-Varshamov} bounds.
This work considers a class of channels over $\V=\F_q^{n\times m}$, where the error set models the phased bursts mentioned in \Cref{sect:intro}. 

\begin{definition}\label{def:channel}
    Let $q$ be a prime power, $n,m,w\in\Z_{\geq0}$ with $w\leq m$, and $\eseta\subseteq\esetb$ be arbitrary subsets of $\F_q^n$ s.t. $\60\in\eseta$.
    The set of $(\eseta,\esetb,w)$-Phased Burst Errors ($(\eseta,\esetb,w)$-PBEs) in $\F_q^{n\times m}$ is the set $\E=\E(n,m,\eseta,\esetb,w)$, where
    \begin{align*}
        \eset=\{\6X\in\F_q^{n\times m}&:\forall \6x\in\col(\6X),\6x\in\esetb\textnormal{ and }\\
        &|\{\6x\in\col(\6X):\6x\in\eseta\}|\geq m-w \}.
    \end{align*}
    The Phased Burst Channel (PBC) $\PBC(n,m,\eseta,\esetb,w)$ is the channel on $\F_q^{n\times m}$ with error set $\eset$.
    A $(\eseta,\esetb,w)$-PBE Correcting Code ($(\eseta,\esetb,w)$-PBECC) is a one-shot code for this channel.
\end{definition}

The second condition in the definition implies that at most $w$ columns of a PBE are in $\esetb\setminus\eseta$ (bad columns), while the others are from $\eseta$ (good columns).
The condition that $\60\in\eseta$ allows for error-free uses of the channel and is, therefore, natural.
The case where $\eseta=\{\60\}$ (the bursts affect up to $w$ columns, while the good columns are error-free) is included in the definition.
No constraint on the position of the bad columns is assumed.

For the setting of \Cref{def:channel}, $\dim(\V)=nm$ in \Cref{eq:one_shot_capacity}, 
and $C_1(\PBC(n,m,\eseta,\esetb,w))=\log_q(A(n,m,\eseta,\esetb,w))/nm$.
Hence, \Cref{quest:codesize} can be equivalently stated as

\begin{question}\label{quest:capacity}
    Given $n$, $m$, $\{0\}\subseteq\eseta\subseteq\esetb\subseteq\F_q^n$ and $w$, what is the one-shot capacity $C_1(\PBC(n,m,\eseta,\esetb,w))$ of the Phased Burst Channel?
\end{question}

While the two questions are equivalent, \Cref{quest:capacity} shifts the focus from the code size to the channel itself, and thus to the error set.
Similar to \Cref{example:Hamming_channel_2}, we are also interested in the asymptotic behavior of $C_1(\PBC(n,m,\eseta,\esetb,w))$.

\begin{definition}\label{def:PBE}
    Let $M>0,\W\in[0,1]$ be fixed constants, then a sequence of PBEC is a sequence
    $\PBC(n,M,\eseta,\esetb,\W)=(\PBC(n,m,\eseta(n),\esetb(n),w))_{n\in\N}$ such that for any $n$ we have $m=Mn$ and $w=\W m$.
    The associated sequence of PBE sets of these channels is denoted by $\eset(n)=\eset(n,M,\eseta,\esetb,W)$.
    We write $\PBC(n)=\PBC(n,M,\eseta,\esetb,\W)$, and $\eset(n)=\eset(n,M,\eseta,\esetb,\W)$ to highlight the dependence on $n$ when everything else is clear from the context.
\end{definition}

Given a sequence of PBEC, we look at the size of the associated PBE sets as $n\rightarrow\infty$ and infer asymptotic bounds on the one-shot capacity of the channels, similarly to the ones in \Cref{example:Hamming_channel_2} for the Hamming channel.
To illustrate the behavior of our general results, we consider Hamming-metric bursts as a particular case of \Cref{def:PBE}.

\begin{definition}[Hamming PBC/PBEs]\label{example:PBE_Hamming_ball}
Let $\wtt\leq n$. 
Then, the Hamming Phased Burst Channel is the channel $\HPBC(n,m,\wtt,w)=\PBC(n,m,\HBall{q}{n}{0},\HBall{q}{n}{\wtt},w)$.
The associated error set is denoted as $\Heset(n,m,\wtt,w)$ (Hamming PBEs, H-PBEs).
\end{definition}

A graphical representation of two elements of a binary Hamming PBE set is given in \Cref{fig:channel_model}: error columns are highlighted in light blue, and blanks are zeros.
Coding problems related to this channel are widely studied in the literature; see \cite{gabrys2012graded,roth2014coding}.
Unlike other authors, we formulate our problem more generally, studying fundamental bounds and constructions while also considering the asymptotic setting where the number of errors is a fraction of the length.

%% file: sections/3_bounds_general.tex
This section analyzes the asymptotic size of sequences of PBE sets $\eset(n)=\eset(n,M,\eseta,\esetb,W)$ (and their corresponding difference sets) to derive bounds on the one-shot capacity of the associated channel sequence $\PBC(n,M,\eseta,\esetb,W)$.

\begin{definition}[Admissible sequences of PBE sets/channels]\label{def:admissible}
Let $M>0$, and $\W,\ca,\cb,\caa,\cab,\cba,\cbb\in[0,1]$ be fixed constants.
A sequence of PBE sets $\eset(n,M,\eseta,\esetb,\W)$ is \emph{admissible} if $|\eset_j(n)|=q^{c_jn\pm o(n)}$ and $|\eset_{ij}(n)|=q^{c_{ij}n\pm o(n)}$, where $\eset_{ij}(n) = \Delta(\eset_i(n),\eset_j(n))$, $i,j=1,2$.
The sequence of channels associated with such a sequence of error sets is an admissible sequence of PBCs.
For brevity, the parameter $n$ is omitted when clear from the context.
\end{definition}

We begin by estimating the asymptotic size of $\eset(n)$.

\begin{theorem}[PBE Hamming bound]\label{thm:PBE_Hamm}
For an admissible sequence of error sets $\eset(n)=\eset(n,M,\eseta,\esetb,\W)$, we have
\[
q^{((1-\W) \ca+\W\cb-o(1))Mn^2}\leq|\eset(n)|\leq q^{((1-\W)\ca+\W\cb+o(1))Mn^2}
\]
and the one-shot capacity of the corresponding sequence of channels $\PBC(n)=\PBC(n,M,\eseta,\esetb,\W)$ is bounded as
\begin{equation*}
    C_1(\PBC(n))\leq R_\mathrm{H}+o(1)\textnormal{, with }R_\mathrm{H}=1-(1-\W )\ca-\W\cb.
\end{equation*}
\end{theorem}

\begin{proof}
Looking at the structure of $\eset(n)$ we find that
\begin{equation*}
    |\eseta(n)|^{m-w}|\esetb(n)|^w\leq|\eset(n)|\leq\binom{m}{w}|\eseta(n)|^{m-w}|\esetb(n)|^w.
\end{equation*}
and using the fact that $\eset(n)$ is admissible the bounds on $|\eset(n)|$ follow.
From the lower bound we obtain
\begin{equation*}
    (1-\W)\ca+\W\cb-o(1)\leq\frac{\log_q(|\eset(n)|)}{mn}
\end{equation*}
by taking the logarithm and normalizing, and the bound on the one-shot capacity follows from \Cref{lem:Hamming_adversarial}.
\end{proof}

\begin{example}
Let $M>0$, $\W\in[0,1]$, and for all $n\in\N$ let $0\leq t_n\leq n$ be such that $\lim_{n\rightarrow\infty}\frac{t_n}{n}=T\in[0,1]$.
Then, the sequence of Hamming PBE sets $\Heset(n,Mn,Tn,\W m)$ (Def.~\ref{example:PBE_Hamming_ball}) is admissible according to Def.~\ref{def:admissible}.
We have $c_1=0$, $c_2=F_q(T)$ (\Cref{eq:general_approximations}), and since the difference set of two Hamming balls is a Hamming ball, $c_{11}=0$, $c_{12}=c_{21}=F_q(T)$, and $c_{22}=F_q(2T)$.
\end{example}
 
\begin{corollary}[H-PBE Hamming bound]\label{corr:HPBE_H}
Let $\eseta = \{\60\}$, $\esetb = \HBall{q}{n}{\wtt}$ with $\wtt = \T\cdot n$.
Then, $R_\mathrm{H} = 1-\W\cdot F_q(\T)$.
\end{corollary}

Having found an upper bound on $|\eset(n)|$ and the maximum possible code rate in the associated channels (Hamming bound), we now analyze $|\Delta(\eset(n))|$ to obtain a lower (Gilbert-Varshamov) bound on $C_1(\PBC(n))$.
 
\begin{theorem}[PBE GV bound\label{thm:GV}]
For an admissible sequence of error sets $\eset(n)=\eset(n,M,\eseta,\esetb,\W)$, it holds that
\begin{equation*}
    \log_q(|\Delta(\eset(n))|)\leq \alpha Mn^2+o(n), 
\end{equation*}
where, if $\caa+\cbb\leq2\cab$, we have
\begin{equation*}
    \alpha=\begin{cases}
        (1-2\W)\caa+2\W\cab,&\textnormal{ if }2\W\leq1,\\
        2(1-\W)\cab+(2\W-1)\cbb,&\textnormal{ if }2\W>1,
    \end{cases}
\end{equation*}
and otherwise $\alpha=(1-\W)\caa+\W\cbb$.
In all cases, let $R_\mathrm{GV} = 1-\alpha$; then the one-shot capacity of the associated sequence of channels is bounded as
\begin{equation*}
    C_1(\PBC(n))\geq R_\mathrm{GV}-o(1).
\end{equation*}
\end{theorem}
\begin{proof}
Any element of $\Delta(\eset(n))$ is of the form $\6X-\6Y$ with $\6X,\6Y\in\eset(n)$.
Let $\6X$ have $x \leq w$ columns in $\esetb$ and $m-x$ columns in $\eseta$.
Let $\6Y$ have $y \leq w$ columns in $\esetb$ and $m-y$ columns in $\eseta$.
Then, $\6X-\6Y$ with $z$ columns in $\esetbb$ has $x-z$ columns in $\esetab$,  $y-z$ columns in $\esetba$, and  $m-x-y-z$ columns in $\esetaa$.
Due to $|\esetab| = |\esetba|$, counting possible column vectors and their permutations, we bound $|\Delta(\eset(n))|$ from above as
\begin{align*}
&\sum_{x,y,z}
\binom{m}{y-z,x-z,z}|\esetbb|^{z} |\esetab|^{x+y-2z} |\esetaa|^{m-x-y+z} \\
= &2^{O(n)} \max_{x,y,z}
|\esetbb|^{z} |\esetab|^{x+y-2z} |\esetaa|^{m-x-y+z},
\end{align*}
for $\max\{0,x+y-m\}\leq z\leq \min\{x,y\}$.
As $|\esetaa|\leq|\esetab|\leq|\esetbb|$, the expression is maximized for $x=y=\wtw$ and we get
\[
|\Delta(\eset(n))| \leq 2^{O(n)} |\esetaa|^{m-2\wtw} \cdot |\esetab|^{2\wtw} \max_z
\frac{|\esetaa|^z\cdot |\esetbb|^{z}}{|\esetab|^{2z}}.
\]
If $|\esetaa|\cdot|\esetbb| \leq |\esetab|^2$, the maximum is obtained for the minimal $z$, i.e., $z = \max\{0,2w-m\}$.
For $2\wtw \leq m$, we get 
\[
\frac{\log_q(|\Delta(\eset(n))|)}{n\cdot m} \leq (1-2\W)\caa + 2\W\cab + o(1)\; ,
\]
and, for $2\wtw > m$, $z=2\wtw-m$ yields
\[
\frac{\log_q(|\Delta(\eset(n))|)}{n\cdot m} \leq  2(1-\W)\cab +(2\W-1)\cbb + o(1)\;.
\]
In case of $|\esetaa|\cdot|\esetbb| > |\esetab|^2$, the maximum is obtained for the maximum possible $z$; that is, $z = \wtw$, which results in
\[
\frac{\log_q(|\Delta(\eset(n))|)}{n\cdot m} \leq  (1-\W) \caa +\W\cbb + o(1)\;.
\]
The bound on the one-shot capacity follows from Thm.~\ref{thm:GV_generic}.
\end{proof}

\begin{remark}\label{rem:both_relevant}
$\caa+\cab\leq2\cab$ or rather $|\esetaa|\cdot |\esetbb| \leq |\esetab|^2$ can be considered the standard case in \Cref{thm:GV}, with many choices of $\eseta\subset\esetb$ falling into this category: 
\begin{itemize}[leftmargin=*]
    \item \textbf{Max-norm}: $\eseta = \{-a,\ldots,a\}^n$,  $\esetb = \{-b,\ldots,b\}^n$, $a\leq b$,
    \item \textbf{Hamming metric}: $\eseta = \HBall{q}{n}{\wtta}$,  $\esetb = \HBall{q}{n}{\wttb}$, $\wtta\leq \wttb$,
    \item \textbf{Lin.~subspaces}: $\eseta = \langle\6e_1,\ldots,\6e_r\rangle$,$\esetb = \langle\6e_1, \ldots,\6e_s\rangle$,$r\leq s$.
\end{itemize}
However, there are also cases where the opposite is true.
Let 
\[
\eseta = \{0,3,7\}^n \subset \esetb = \{-4,0,3,7,10\}^n.
\]
Then, for $q$ sufficiently large, $|\esetaa|=7^n$, $|\esetab| = 9^n$, and $|\esetbb|=13^n$.
Hence, $|\esetaa|\cdot |\esetbb|=91^n$ and $|\esetab|=81^n$.
\end{remark}

For Hamming PBEs, $\caa + \cbb\leq2\cab$ is satisfied, and we obtain the following expressions.

\begin{corollary}[H-PBE GV bound]\label{corr:HPBE_GV}
Let $\eseta = \{\60\}$, and $\esetb = \HBall{q}{n}{\wtt}$ with $\wtt = \T\cdot n$.
For $2\W\leq 1$, the rate of any $(t,w)$-PBECC is lower bounded by $R_\mathrm{GV} = 1 - 2\W \cdot F_q(\T)$;
else, $R_\mathrm{GV} =1- 2(1-\W)F_q(\T) - (2\W-1)F_q(2\T)$.
\end{corollary}

Every $(w\cdot t)$-error-correcting code in $\F_q^{nm}$ is trivially a ($t$,$w$)-PBEC, and these codes are subject to the Hamming and GV bounds in \Cref{example:Hamming_channel_2}.
It is natural to ask how these rates compare with the presented bounds for Hamming-metric PBEs.

\begin{comparison}\label{comparison:PBE_vs_Hamm}
Let $\RH$ and $R_\mathrm{GV}$ be the quantities defined in \Cref{corr:HPBE_H,corr:HPBE_GV}. 
Comparing with the asymptotic Hamming and GV bounds for block codes (\Cref{example:Hamming_channel_2}), we find
\[
\RH \geq 1- F_q(\W\T)\textnormal{ and } R_\mathrm{GV}\geq 1 - F_q(2\W\T),
\]
implying that the upper and lower bounds on the rate increase asymptotically when accounting for the structure of the PBEs, instead of blindly correcting any $\wtt\wtw=\T\W \cdot nm$ errors.
\end{comparison}

Finally, we note that there are instances where the rate $R_\mathrm{GV}$ of Cor.~\ref{corr:HPBE_GV} cannot be achieved by a $(\wtt\cdot \wtw)$-error-correcting code.

\begin{example}\label{ex:HammingPBE_vs_ClassicalHamming}
Let $q=2$, $\T = \tfrac{1}{5}$, and $\W=\tfrac{1}{12}$. 
Then, by the classical Hamming bound, every code correcting a fraction of $\W\cdot\T =\tfrac{1}{60}$ errors has a rate of at most $0.878$.
For Hamming PBEs, \Cref{corr:HPBE_GV} implies that rate $0.880$ is achievable.
\end{example}

%% file: sections/4_construction_general.tex
It is natural to align the structure of a code construction with the structure of the errors it is required to correct.
Generalized Concatenated Codes (GCCs) follow this principle, coding both within and over different columns \cite{zinov1976generalized}.
For particular cases, the correction of burst error was studied in \cite{zinov1981generalized, zinov1979correction}.
We refer to \cite{dumer1998concatenated} for a comprehensive introduction.

In the following, $[n,k, \deset]_q$ denotes a $k$-dimensional linear code in $\F_q^n$ that can correct all error patterns $\eset$.
For $\eset = \HBall{q}{n}{t}$, we simply write $[n,k, 2t+1]_q$ as is common.

\begin{construction}[GCC \cite{zinov1976generalized}]\label{const:gcc}
Let inner codes $\codeB_s\subset \cdots \subset \codeB_1$ be given with $\codeB_j = [n,k_j,d_j]_q$. 
Let outer codes $\codeA_1, \ldots, \codeA_s$ be given with $\codeA_j = [m,K_j, D_j]_{q^{(k_j - k_{j+1})}}$ for $j<s$ and $\codeA_s = [m,K_s, D_s]_{q^{k_s}}$.
Denote as $\oplus$ the direct sum and as $\otimes$ the (tensor) concatenation.
Then, an $s$-level GCC is the code
\[
\code(\codeA_{[s]},\codeB_{[s]})\coloneqq \bigoplus_{j=1}^{s-1}\big(\codeA_j \otimes \left(\codeB_{j}/ \codeB_{j+1}\right) \big) \oplus \left(\codeA_s \otimes \codeB_s\right) \subset \F_q^{n\times m}.
\]
\end{construction}

\begin{example}\label{ex:gcc}
Pick $\codeB_1=[4,3,2]_2 \subset \codeB_2 = [4,1,4]_2$ and $\codeA_1 =\langle (1,1)\rangle_{\F_4}$, $\codeA_2 = \F_2^2$.
Then, $\codeB_1 / \codeB_1 = \langle(1\,1\,0\,0),(1\,0\,1\,0)\rangle$ and 
\[
\code(\codeA_{[2]},\codeB_{[2]}) =\left\langle
\begin{pmatrix}
1 & 0 \\
1 & 0 \\
1 & 0 \\
1 & 0 \\
\end{pmatrix}\!,
\begin{pmatrix}
0 & 1 \\
0 & 1 \\
0 & 1 \\
0 & 1 \\
\end{pmatrix}\!,
\begin{pmatrix}
0 & 0 \\
1 & 1 \\
0 & 0 \\
1 & 1 \\
\end{pmatrix}\!,
\begin{pmatrix}
0 & 0 \\
0 & 0 \\
1 & 1 \\
1 & 1 \\
\end{pmatrix}
\right\rangle. 
\]
\end{example}

Similar to the well-known lower bound on the minimum distance, we can make the following general observation about when \Cref{const:gcc} is $(\eseta,\esetb,w)$-PBEC.

\begin{property}\label{gcc:ecc}
The code $\code(\codeA_{[s]},\codeB_{[s]})$ constructed in \Cref{const:gcc} has dimension $\sum_{j=1}^s k_j K_j$.
$\code(\codeA_{[s]},\codeB_{[s]})$ is guaranteed to be $(\eseta,\esetb,w)$-PBEC if for each level $j\in[s]$
\begin{align*}
            D_j > 2w&\text{ and }\codeB_j \cap \esetaa =\{\60\},\\
\text{or\hspace{2em}}  D_j > \hphantom{2}w&\text{ and }\codeB_j \cap \esetab =\{\60\},\\
\text{or\hspace{2em}}  \hphantom{D_j > 2w}&\hphantom{\text{ and }}\codeB_j \cap \esetbb =\{\60\}.\\
\end{align*}
\end{property}

\begin{proof}
The dimension is computed as usual.
$\codeC(\codeA_{[s]}, \codeB_{[s]})$ is $(\eseta,\esetb,w)$-PBEC iff $\deset\cup\codeC(\codeA_{[s]}, \codeB_{[s]}) = \{ \60 \}$.
This is the case if $\deset \cup \big(\codeA_j \otimes \left(\codeB_{j}/ \codeB_{j+1}\right) \big) =\{\60\}$ $\forall j\in[s-1]$ and $\deset \cup \big(\codeA_s \otimes \codeB_{s} \big) =\{\60\}$.
Let $\6C$ be nonzero codeword of $\big(\codeA_j \otimes \left(\codeB_{j}/ \codeB_{j+1}\right) \big)$.
Then, $\6C$ has at least $D_j$ nonzero columns.
\begin{description}
    \item[$\codeB_j \cap \esetbb =\{\60\}$:] 
    Since no element of $\deset$ has a column in $\F_q^n\setminus\esetbb$, $D_j > 0 $ is sufficient to guarantee that $\6C\notin\deset$.
    \item[$\codeB_j \cap \esetab =\{\60\}$:] 
    At most $w$ columns of an element in $\deset$ are in $\F_q^n\setminus(\esetab\cup\esetba)$.
    As $\codeB_j$ is linear, $\codeB_j \cap \esetab = \{\60\} \iff \codeB_j \cap (\esetab\cup\esetba) = \{\60\}$.
    Hence, each nonzero column of $\6C$ is in $\F_q^n\setminus(\esetab\cup\esetba)$,
    and $D_j > w$ implies $\6C\notin\deset$.
    \item[$\codeB_j \cap \esetaa =\{\60\}$:]
    At most $2w$ columns of an element in $\deset$ are in $\F_q^n\setminus\esetaa$.
    Since each nonzero column of $\6C$ is in $\F_q^n\setminus\esetaa$,
    $D_j > 2w$ guarantees that $\6C\notin\deset$.
\end{description}
The case $\6c \in \big(\codeA_s \otimes \codeB_{s} \big) \setminus\{\60\}$ follows in the same way.
\end{proof}

We illustrate \Cref{gcc:ecc} in the case of Hamming PBEs.

\begin{corollary}\label{gcc:Hecc}
$\code(\codeA_{[s]},\codeB_{[s]})$ is $(t,w)$-PBEC if $\forall j\in[s]$
\begin{align*}
            D_j > 2w&,\\
\text{or\hspace{2em}}  D_j > \hphantom{2}w&\text{ and }d_j > \hphantom{2}t,\\
\text{or\hspace{2em}}  \hphantom{D_j > 2w}&\hphantom{\text{ and }}d_j > 2t.\\
\end{align*}
\end{corollary}

\begin{example}
Applying \Cref{gcc:Hecc}, we observe that the code  $\code(\codeA_{[2]},\codeB_{[2]})$ constructed in \Cref{ex:gcc} is $(1,1)$-PBEC.
\end{example}

Inspired by \Cref{gcc:ecc} and the construction of \cite{wolf1965codes}, we systematically construct $(\eseta, \esetb,w)$-PBECCs via the following choice of inner and outer codes.

\begin{construction}[$2$-level code]\label{const:2lvl}
Select
$\codeB_1 = [n,k_1,\esetaa]_q$, 
$\codeB_2 = [n,k_2,\esetbb]_q$ as inner,
and 
$\codeA_1 = [m, K_1, 2w+1]_{q^{k_1-k_2}}$, 
$\codeA_2 = \F^m_{q^{k_2}}$ as outer codes.
By \Cref{gcc:ecc}, $\code(\codeA_{[2]},\codeB_{[2]})$ is $(\eseta,\esetb,w)$-PBEC.
We select $\codeB_1, \codeB_2$ on the GV bound and, for sufficiently%
\footnote{$|\esetbb|\cdot q^{-o(n)} \geq m |\esetaa|$ is sufficient, and often holds for moderate $n$.} %
large $n$, $\codeA_1$ MDS.
Let $|\esetaa| = q^{\caa\cdot n + o(n)}$, and $|\esetbb| = q^{\cbb\cdot n + o(n)}$.
Then, the overall rate is $\RTwo-o(1)$ with 
\[
\RTwo = 1 - \cbb + (\cbb - \caa)\max\{1-2W,0\}.
\]
\end{construction}

Indeed, Hamming phased bursts recover the code of \cite{wolf1965codes}.

\begin{corollary}[$2$-level H-PBE code]\label{corr:H_2lvl}
Let $\wtt=\T\cdot n$, $\wtw =\W\cdot n$.
For $\eseta=\{\60\}$, $\esetb=\HBall{q}{n}{\wtt}$, \Cref{const:2lvl} yields a $(t,w)$-PBECC of rate $R_\mathrm{2lvl} = 1 - \min\{1,2\W\} \cdot F_q(2\T)$, which is equivalent to the construction in \cite{wolf1965codes}.
\end{corollary}

\begin{example}\label{ex:2lvl_vs_GV}
Let $q=2$, $n = m$, $\T = 0.1$, and $\W = 0.2$.
Then, the GV bound is $R_\mathrm{GV} = 0.81$.
Construction~\ref{const:2lvl} achieves $R_\mathrm{2lvl} = 0.71$ for $\codeB_2$ on the GV bound and $\codeA_1$ MDS.
\end{example}

\Cref{const:2lvl} utilizes only two out of three conditions of \Cref{gcc:ecc}, which allow a direct guarantee on the error-correction capability.
This motivates the following generalization.

\begin{construction}[$3$-level code]\label{const:3lvl}
Select 
$\codeB_1 = [n,k_1,\esetaa]_q$, 
$\codeB_2 = [n,k_2,\esetab]_q$, 
$\codeB_3 = [n,k_3,\esetbb]_q$ as inner codes,  
and 
$\codeA_1 = [m, K_1, 2w+1]_{q^{k_1-k_2}}$, 
$\codeA_2 = [m, K_2,  w+1]_{q^{k_2-k_3}}$, 
$\codeA_3 = \F^m_{q^{k_3}}$ as outer codes.
By \Cref{gcc:ecc}, $\code(\codeA_{[3]},\codeB_{[3]})$ is $(\eseta,\esetb,w$)-PBEC.
For sufficiently\footnote{Here, $|\esetbb| \cdot q^{-o(n)}\geq |\esetab| m$ and $|\esetab| \cdot q^{-o(n)} \geq |\esetaa| m$ is sufficient.} %
large $n$, select the inner codes on the GV bound and the outer codes MDS.
Let $|\esetaa| = q^{\caa\cdot n + o(n)}$, $|\esetab| = q^{\cab\cdot n + o(n)}$, and $|\esetbb| = q^{\cbb\cdot n + o(n)}$.
Then, the total rate is $\RThree -o(1)$ with
\[
R_\mathrm{3lvl} = 
\begin{cases}
1 - W(\cab+\cbb) -\caa(1-2\W),& \text{if } 2\W \leq 1,\\
1 - \cab(1-\W) - \cbb\W, & \text{otherwise.}\\
\end{cases}
\]
\end{construction}

For Hamming PBEs, the following corollary is obtained.
\begin{corollary}[$3$-level H-PBE code]\label{corr:H_3lvl}
For $\wtw = \W n$, $\wtt=\T n$, $\eseta=\{\60\}$, $\esetb=\HBall{q}{n}{t}$, \Cref{const:3lvl} yields a $(t,w)$-PBECC of rate $R_\mathrm{3lvl} = 1 - \W  F_q(2\T) -\min\{\W,1-\W\}F_q(\T)$.
\end{corollary}

\begin{example}\label{ex:3lvl_vs_GV}
For $q=2$, $n = m$, $\T = 0.1$, and $\W = 0.2$, \Cref{const:3lvl} achieves $R_\mathrm{3lvl} = 0.76$ using inner codes on the GV bound and outer codes MDS.
This represents a significant rate improvement compared to the $2$-level construction (\Cref{ex:2lvl_vs_GV}).
\end{example}

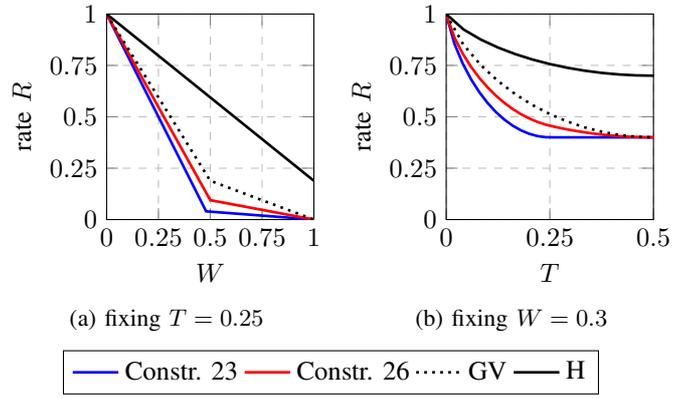
\begin{figure}
    \centering
    \begin{subfigure}[b]{0.25\textwidth}
        \centering
        \input{figures/codes_fixT025}
        \caption{fixing $T = 0.25$}\label{fig:code_fixT}
    \end{subfigure}%
    \begin{subfigure}[b]{0.25\textwidth}
        \centering
        \input{figures/codes_fixW030}
        \caption{fixing $W=0.3$}\label{fig:code_fixW}
    \end{subfigure}
    \vspace{-0.2cm}
    
\pgfplotslegendfromname{myboundlegend}
\caption{Comparison of the rates achieved by \Cref{const:2lvl} and \Cref{const:3lvl} with the GV and Hamming bound for $q=2$.}
\label{fig:code_comparison}
\end{figure}
The achievable rates of Constructions~\ref{const:2lvl} and \ref{const:3lvl} are plotted in \Cref{fig:code_comparison} for further parametrizations of the Hamming PBE channel.
\Cref{const:3lvl} improves over \Cref{const:2lvl} for all shown values of $\W, \T$.
Next, a general formal comparison is provided.

\begin{comparison}\label{comparison:2lvl_3lvl_gv}
The rates of \Cref{const:2lvl} and \Cref{const:3lvl} satisfy
$R_\mathrm{3lvl} = R_\mathrm{2lvl} +\min\{\W,1-\W\}(\cbb-\cab)$. 
Similarly, a comparison with the GV bound (\Cref{thm:GV}) yields
\[
\frac{R_{\mathrm{GV}} - R_\mathrm{3lvl}}{\min\{\W,1-\W\}}= 
\begin{cases}
\cbb - \cab, & \text{if } \caa + \cbb\leq 2\cab,\\
\cab - \caa, & \text{otherwise.}
\end{cases}
\]
That is, the GV bound is achieved for $\caa=\cab$, $\cab=\cbb$. 
\end{comparison}

\Cref{comparison:2lvl_3lvl_gv} shows that the presented code constructions generally do not achieve the GV bound.
On the other hand, efficient decoding of the constructed GCCs is feasible, provided that efficient decoders for the component codes are available.

%% file: figures/codes_fixT025.tex
\begin{tikzpicture}[baseline] 
\begin{axis}[%
ymax = 1,
ymin = 0,
xmin = 0.0,
xmax = 1,
xtick={0, 0.25, 0.5, 0.75, 1}, 
ytick={0, 0.25, 0.5, 0.75, 1}, 
label style={font=\footnotesize},
ylabel={rate $R$},
xlabel={$W$},
ymajorgrids,
xmajorgrids,
grid style=dashed,
ylabel near ticks,
xlabel near ticks,
width = 0.95*\linewidth,
height = 0.95*\linewidth,
legend columns=4,
legend entries={classical GV, classical H, GV, H},
legend to name={myboundlegend},
]

\addplot[line width=1.0pt, blue]
  table[row sep=crcr]{
0.000000000000000 1.0 \\
0.0200000000000000 0.96 \\
0.0400000000000000 0.92 \\
0.0600000000000000 0.88 \\
0.0800000000000000 0.84 \\
0.100000000000000 0.8 \\
0.120000000000000 0.76 \\
0.140000000000000 0.72 \\
0.160000000000000 0.6799999999999999 \\
0.180000000000000 0.64 \\
0.200000000000000 0.6 \\
0.220000000000000 0.56 \\
0.240000000000000 0.52 \\
0.260000000000000 0.48 \\
0.280000000000000 0.43999999999999995 \\
0.300000000000000 0.4 \\
0.320000000000000 0.36 \\
0.340000000000000 0.31999999999999995 \\
0.360000000000000 0.28 \\
0.380000000000000 0.24 \\
0.400000000000000 0.19999999999999996 \\
0.420000000000000 0.16000000000000003 \\
0.440000000000000 0.12 \\
0.460000000000000 0.07999999999999996 \\
0.480000000000000 0.040000000000000036 \\
1 0 \\
};
\addlegendentry{\Cref{const:2lvl}}

\addplot[line width=1.02pt, red]
  table[row sep=crcr]{
0.000000000000000 1.0 \\
0.0200000000000000 0.9637744375108174 \\
0.0400000000000000 0.9275488750216346 \\
0.0600000000000000 0.891323312532452 \\
0.0800000000000000 0.8550977500432694 \\
0.100000000000000 0.8188721875540868 \\
0.120000000000000 0.782646625064904 \\
0.140000000000000 0.7464210625757214 \\
0.160000000000000 0.7101955000865388 \\
0.180000000000000 0.6739699375973561 \\
0.200000000000000 0.6377443751081735 \\
0.220000000000000 0.6015188126189908 \\
0.240000000000000 0.5652932501298081 \\
0.260000000000000 0.5290676876406255 \\
0.280000000000000 0.49284212515144277 \\
0.300000000000000 0.45661656266226014 \\
0.320000000000000 0.4203910001730774 \\
0.340000000000000 0.3841654376838947 \\
0.360000000000000 0.3479398751947122 \\
0.380000000000000 0.3117143127055295 \\
0.400000000000000 0.27548875021634683 \\
0.420000000000000 0.23926318772716432 \\
0.440000000000000 0.20303762523798158 \\
0.460000000000000 0.1668120627487989 \\
0.480000000000000 0.13058650025961627 \\
0.500000000000000 0.09436093777043358 \\
0.520000000000000 0.09058650025961623 \\
0.540000000000000 0.08681206274879888 \\
0.560000000000000 0.08303762523798153 \\
0.580000000000000 0.07926318772716423 \\
0.600000000000000 0.07548875021634688 \\
0.620000000000000 0.07171431270552953 \\
0.640000000000000 0.06793987519471217 \\
0.660000000000000 0.06416543768389482 \\
0.680000000000000 0.06039100017307747 \\
0.700000000000000 0.05661656266226017 \\
0.720000000000000 0.05284212515144282 \\
0.740000000000000 0.04906768764062547 \\
0.760000000000000 0.045293250129808116 \\
0.780000000000000 0.041518812618990764 \\
0.800000000000000 0.03774437510817341 \\
0.820000000000000 0.03396993759735609 \\
0.840000000000000 0.030195500086538762 \\
0.860000000000000 0.02642106257572141 \\
0.880000000000000 0.022646625064904058 \\
0.900000000000000 0.018872187554086706 \\
0.920000000000000 0.015097750043269367 \\
0.940000000000000 0.011323312532452043 \\
0.960000000000000 0.007548875021634691 \\
0.980000000000000 0.0037744375108173453 \\
1.00000000000000 0.0 \\
};
\addlegendentry{\Cref{const:3lvl}}

\addplot[line width=1.0pt, dotted]
  table[row sep=crcr]{
0.000000000000000 1.0 \\
0.0200000000000000 0.9675488750216347 \\
0.0400000000000000 0.9350977500432693 \\
0.0600000000000000 0.9026466250649041 \\
0.0800000000000000 0.8701955000865387 \\
0.100000000000000 0.8377443751081735 \\
0.120000000000000 0.8052932501298081 \\
0.140000000000000 0.7728421251514428 \\
0.160000000000000 0.7403910001730775 \\
0.180000000000000 0.7079398751947121 \\
0.200000000000000 0.6754887502163469 \\
0.220000000000000 0.6430376252379815 \\
0.240000000000000 0.6105865002596162 \\
0.260000000000000 0.5781353752812509 \\
0.280000000000000 0.5456842503028856 \\
0.300000000000000 0.5132331253245204 \\
0.320000000000000 0.4807820003461549 \\
0.340000000000000 0.4483308753677896 \\
0.360000000000000 0.4158797503894244 \\
0.380000000000000 0.38342862541105904 \\
0.400000000000000 0.3509775004326937 \\
0.420000000000000 0.3185263754543285 \\
0.440000000000000 0.28607525047596305 \\
0.460000000000000 0.2536241254975977 \\
0.480000000000000 0.2211730005192325 \\
0.500000000000000 0.18872187554086717 \\
0.520000000000000 0.18117300051923246 \\
0.540000000000000 0.17362412549759776 \\
0.560000000000000 0.16607525047596305 \\
0.580000000000000 0.15852637545432846 \\
0.600000000000000 0.15097750043269376 \\
0.620000000000000 0.14342862541105905 \\
0.640000000000000 0.13587975038942435 \\
0.660000000000000 0.12833087536778964 \\
0.680000000000000 0.12078200034615494 \\
0.700000000000000 0.11323312532452035 \\
0.720000000000000 0.10568425030288564 \\
0.740000000000000 0.09813537528125094 \\
0.760000000000000 0.09058650025961623 \\
0.780000000000000 0.08303762523798153 \\
0.800000000000000 0.07548875021634682 \\
0.820000000000000 0.06793987519471223 \\
0.840000000000000 0.060391000173077525 \\
0.860000000000000 0.05284212515144282 \\
0.880000000000000 0.045293250129808116 \\
0.900000000000000 0.03774437510817341 \\
0.920000000000000 0.030195500086538707 \\
0.940000000000000 0.022646625064904113 \\
0.960000000000000 0.015097750043269409 \\
0.980000000000000 0.0075488750216347045 \\
1.00000000000000 0.0 \\
};
\addlegendentry{GV}

\addplot[line width=1.0pt]
  table[row sep=crcr]{
0.000000000000000 1.0 \\
0.0200000000000000 0.9837744375108174 \\
0.0400000000000000 0.9675488750216347 \\
0.0600000000000000 0.9513233125324521 \\
0.0800000000000000 0.9350977500432693 \\
0.100000000000000 0.9188721875540867 \\
0.120000000000000 0.9026466250649041 \\
0.140000000000000 0.8864210625757214 \\
0.160000000000000 0.8701955000865387 \\
0.180000000000000 0.8539699375973561 \\
0.200000000000000 0.8377443751081735 \\
0.220000000000000 0.8215188126189907 \\
0.240000000000000 0.8052932501298081 \\
0.260000000000000 0.7890676876406255 \\
0.280000000000000 0.7728421251514428 \\
0.300000000000000 0.7566165626622602 \\
0.320000000000000 0.7403910001730775 \\
0.340000000000000 0.7241654376838949 \\
0.360000000000000 0.7079398751947121 \\
0.380000000000000 0.6917143127055295 \\
0.400000000000000 0.6754887502163469 \\
0.420000000000000 0.6592631877271642 \\
0.440000000000000 0.6430376252379815 \\
0.460000000000000 0.6268120627487989 \\
0.480000000000000 0.6105865002596162 \\
0.500000000000000 0.5943609377704335 \\
0.520000000000000 0.5781353752812509 \\
0.540000000000000 0.5619098127920683 \\
0.560000000000000 0.5456842503028856 \\
0.580000000000000 0.529458687813703 \\
0.600000000000000 0.5132331253245204 \\
0.620000000000000 0.49700756283533765 \\
0.640000000000000 0.4807820003461549 \\
0.660000000000000 0.4645564378569723 \\
0.680000000000000 0.4483308753677896 \\
0.700000000000000 0.4321053128786071 \\
0.720000000000000 0.4158797503894244 \\
0.740000000000000 0.39965418790024176 \\
0.760000000000000 0.38342862541105904 \\
0.780000000000000 0.3672030629218763 \\
0.800000000000000 0.3509775004326937 \\
0.820000000000000 0.3347519379435111 \\
0.840000000000000 0.3185263754543285 \\
0.860000000000000 0.30230081296514577 \\
0.880000000000000 0.28607525047596305 \\
0.900000000000000 0.26984968798678044 \\
0.920000000000000 0.2536241254975977 \\
0.940000000000000 0.23739856300841522 \\
0.960000000000000 0.2211730005192325 \\
0.980000000000000 0.2049474380300499 \\
1.00000000000000 0.18872187554086717 \\
};
\addlegendentry{H}

\end{axis}
\end{tikzpicture}

%% file: figures/codes_fixW030.tex
\begin{tikzpicture}[baseline] 
\begin{axis}[%
ymax = 1,
ymin = 0,
xmin = 0.0,
xmax = 0.5,
xtick={0, 0.25, 0.5, 0.75, 1}, 
ytick={0, 0.25, 0.5, 0.75, 1}, 
label style={font=\footnotesize},
ylabel={rate $R$},
xlabel={$T$},
ymajorgrids,
xmajorgrids,
grid style=dashed,
ylabel near ticks,
xlabel near ticks,
width = 0.95*\linewidth,
height = 0.95*\linewidth,
]

\addplot[line width=1.0pt, blue]
  table[row sep=crcr]{
0.000000000000000 1.00000000000000 \\
0.0200000000000000 0.8546246865505511 \\
0.0400000000000000 0.7586924858786364 \\
0.0600000000000000 0.6823834808275814 \\
0.0800000000000000 0.6194142672156604 \\
0.100000000000000 0.5668431430675825 \\
0.120000000000000 0.5229758323692867 \\
0.140000000000000 0.4867295136639217 \\
0.160000000000000 0.45737112536530367 \\
0.180000000000000 0.4343900864467046 \\
0.200000000000000 0.4174296433271988 \\
0.220000000000000 0.4062474872667666 \\
0.240000000000000 0.4006926784028789 \\
0.260000000000000 0.400000000000000 \\
0.280000000000000 0.400000000000000 \\
0.300000000000000 0.400000000000000 \\
0.320000000000000 0.400000000000000 \\
0.340000000000000 0.400000000000000 \\
0.360000000000000 0.400000000000000 \\
0.380000000000000 0.400000000000000 \\
0.400000000000000 0.400000000000000 \\
0.420000000000000 0.400000000000000 \\
0.440000000000000 0.400000000000000 \\
0.460000000000000 0.400000000000000 \\
0.480000000000000 0.400000000000000 \\
0.500000000000000 0.400000000000000 \\
0.520000000000000 0.400000000000000 \\
0.540000000000000 0.400000000000000 \\
0.560000000000000 0.400000000000000 \\
0.580000000000000 0.400000000000000 \\
0.600000000000000 0.400000000000000 \\
0.620000000000000 0.400000000000000 \\
0.640000000000000 0.400000000000000 \\
0.660000000000000 0.400000000000000 \\
0.680000000000000 0.400000000000000 \\
0.700000000000000 0.400000000000000 \\
0.720000000000000 0.400000000000000 \\
0.740000000000000 0.400000000000000 \\
0.760000000000000 0.400000000000000 \\
0.780000000000000 0.400000000000000 \\
0.800000000000000 0.400000000000000 \\
0.820000000000000 0.400000000000000 \\
0.840000000000000 0.400000000000000 \\
0.860000000000000 0.400000000000000 \\
0.880000000000000 0.400000000000000 \\
0.900000000000000 0.400000000000000 \\
0.920000000000000 0.400000000000000 \\
0.940000000000000 0.400000000000000 \\
0.960000000000000 0.400000000000000 \\
0.980000000000000 0.400000000000000 \\
};

\addplot[line width=1.0pt, red]
  table[row sep=crcr]{
0.000000000000000 1.00000000000000 \\
0.0200000000000000 0.8848801805127293 \\
0.0400000000000000 0.8066585862145937 \\
0.0600000000000000 0.7429582646674479 \\
0.0800000000000000 0.6890533765471484 \\
0.100000000000000 0.642722893457007 \\
0.120000000000000 0.602679656598434 \\
0.140000000000000 0.5680931133391041 \\
0.160000000000000 0.538392696290482 \\
0.180000000000000 0.5131719295048683 \\
0.200000000000000 0.4921363931973907 \\
0.220000000000000 0.4750734927447936 \\
0.240000000000000 0.46183425538608286 \\
0.260000000000000 0.4519760882522146 \\
0.280000000000000 0.4433647568319608 \\
0.300000000000000 0.43561273023079217 \\
0.320000000000000 0.4286855626826518 \\
0.340000000000000 0.4225543885080909 \\
0.360000000000000 0.41719504322335227 \\
0.380000000000000 0.4125873933321101 \\
0.400000000000000 0.40871482166359935 \\
0.420000000000000 0.4055638314899039 \\
0.440000000000000 0.40312374363338327 \\
0.460000000000000 0.4013864683539322 \\
0.480000000000000 0.4003463392014394 \\
0.500000000000000 0.39999999999999997 \\
0.520000000000000 0.400000000000000 \\
0.540000000000000 0.400000000000000 \\
0.560000000000000 0.400000000000000 \\
0.580000000000000 0.400000000000000 \\
0.600000000000000 0.400000000000000 \\
0.620000000000000 0.400000000000000 \\
0.640000000000000 0.400000000000000 \\
0.660000000000000 0.400000000000000 \\
0.680000000000000 0.400000000000000 \\
0.700000000000000 0.400000000000000 \\
0.720000000000000 0.400000000000000 \\
0.740000000000000 0.400000000000000 \\
0.760000000000000 0.400000000000000 \\
0.780000000000000 0.400000000000000 \\
0.800000000000000 0.400000000000000 \\
0.820000000000000 0.400000000000000 \\
0.840000000000000 0.400000000000000 \\
0.860000000000000 0.400000000000000 \\
0.880000000000000 0.400000000000000 \\
0.900000000000000 0.400000000000000 \\
0.920000000000000 0.400000000000000 \\
0.940000000000000 0.400000000000000 \\
0.960000000000000 0.400000000000000 \\
0.980000000000000 0.400000000000000 \\
};

\addplot[line width=1.00pt, dotted]
  table[row sep=crcr]{
0.0 1.00000000000000\\
0.041666666666666664 0.850070624300089\\
0.08333333333333333 0.751709889817820\\
0.125 0.673861334080242\\
0.16666666666666666 0.609986547010988\\
0.20833333333333331 0.557029080313930\\
0.25 0.513233125324520\\
0.29166666666666663 0.477481318458781\\
0.3333333333333333 0.449022499567306\\
0.375 0.427339598245021\\
0.41666666666666663 0.412078746009308\\
0.4583333333333333 0.403009103088418\\
0.5 0.400000000000000\\
0.5416666666666666 0.400000000000000\\
0.5833333333333333 0.400000000000000\\
0.625 0.400000000000000\\
0.6666666666666666 0.400000000000000\\
0.7083333333333333 0.400000000000000\\
0.75 0.400000000000000\\
0.7916666666666666 0.400000000000000\\
0.8333333333333333 0.400000000000000\\
0.875 0.400000000000000\\
0.9166666666666666 0.400000000000000\\
0.9583333333333333 0.400000000000000\\
1.0 0.400000000000000\\
};

\addplot[line width=1.0pt]
  table[row sep=crcr]{
0.0 1.00000000000000\\
0.041666666666666664 0.925035312150044\\
0.08333333333333333 0.875854944908910\\
0.125 0.836930667040121\\
0.16666666666666666 0.804993273505494\\
0.20833333333333331 0.778514540156965\\
0.25 0.756616562662260\\
0.29166666666666663 0.738740659229391\\
0.3333333333333333 0.724511249783653\\
0.375 0.713669799122510\\
0.41666666666666663 0.706039373004654\\
0.4583333333333333 0.701504551544209\\
0.5 0.700000000000000\\
0.5416666666666666 0.700000000000000\\
0.5833333333333333 0.700000000000000\\
0.625 0.700000000000000\\
0.6666666666666666 0.700000000000000\\
0.7083333333333333 0.700000000000000\\
0.75 0.700000000000000\\
0.7916666666666666 0.700000000000000\\
0.8333333333333333 0.700000000000000\\
0.875 0.700000000000000\\
0.9166666666666666 0.700000000000000\\
0.9583333333333333 0.700000000000000\\
1.0 0.700000000000000\\
};

\end{axis}
\end{tikzpicture}

%% file: sections/6_appendix_general.tex
\subsection{Proofs Related to the Proposed Bounds}

\paragraph{More admissible sequences of error sets and channels}

Throughout this paper, the behavior of our general bounds and constructions is illustrated via Hamming PBEs.
In the following, we demonstrate that plenty of other error sets that appear naturally are admissible according to \Cref{def:admissible}, expanding upon \Cref{rem:both_relevant}.

\textbf{Max-norm error sets.}
Let $\eseta(n) = \{-a,-a+1,\ldots,-a\}^n$ and $\esetb(n) =\{-b,-b+1,\ldots,b\}^n$ with fixed $a\leq b \leq \tfrac{q-1}{2}$.
The coefficients for the sizes of $\eseta$ and $\esetb$ are computed as $\ca = \log_q(2a+1)$ and $\cb = \log_q(2b+1)$.
The corresponding difference sets are $\esetaa = \{-2a,\ldots, 2a\}^n$, $\esetab = \esetba = \{-a-b,\ldots,a+b\}^n$ and $\esetbb = \{-2b,\ldots,2b\}^n$
with $\caa = \log_q(\max\{q,4a+1\})$, $\cab = \log_q(\max\{q,2a+2b+1\})$ and $\caa = \log_q(\max\{q,4b+1\})$.

More generally, any error set sequences of the form $\eseta(n) = (\tilde{\eseta})^n$, $\esetb(n) = (\tilde{\esetb})^n$ with fixed $0\in\tilde{\eseta} \subseteq \tilde{\esetb}\subseteq \F_q$ are admissible according to \Cref{def:admissible}.

\textbf{Linear subspaces.}
For $s = S\cdot n$, $t = T\cdot n$ and $S\leq T$, let $\6e_1, \ldots,\6e_s,\ldots\6e_t \in \F_q^n$ be linearly independent.
We set $\eseta(n) = \langle\6e_1,\ldots\6e_s\rangle$ and $\esetb(n) =\langle\6e_1,\ldots,\6e_s,\ldots,\6e_t\rangle$.
The coefficients for the sizes of $\eseta$ and $\esetb$ are computed as $\ca = S$ and $\cb = T$.
The corresponding difference sets are $\esetaa = \eseta$, $\esetab = \esetba =\esetbb = \esetb$.
Hence, the cardinality coefficients are $\caa = S$ and $\cab = \cba = \cbb = T$.

\paragraph{Proof of \Cref{corr:HPBE_H,corr:HPBE_GV}}

Let $\eseta = \{\60\}$ and $\esetb = \HBall{q}{n}{t}$ with $t=\T\cdot n$.
Then, according to \Cref{eq:general_approximations}, $\ca = 0$ and $\cb = F_q(T)$.
$\eseta = \{\60\}$ implies $\Delta(\eseta,\eseta) = \{0\}$ and $\Delta(\eseta,\esetb) = \Delta(\esetb,\eseta) = \esetb$. 
Further, $\Delta(\esetb,\esetb) = \HBall{q}{n}{\min\{2\wtt,n\}}$ because any vector of weight at most $2\wtt$ can be represented as the sum of two vectors of weight at most $t$.
Hence, $\caa = 0$, $\cab = F_q(T)$, and $\cbb = F_q(2T)$.
Then, \Cref{corr:HPBE_H,corr:HPBE_GV} follow from \Cref{thm:PBE_Hamm,thm:GV}.

\paragraph{Proof of \Cref{comparison:PBE_vs_Hamm}}

Define $\eset_\mathrm{H} \coloneqq \{\6E\in \F_q^{n\times m}:\wtH(\6E) \leq w\cdot t\}$
Then, $\eset_\mathrm{H} \supset \eset(n,m,\HBall{q}{n}{0},\HBall{q}{n}{\wtt},w)$, which implies the statement on the corresponding Hamming bounds.
The comparison of the GV bounds follows from $\Delta(\eset_\mathrm{H}) \supset \deset$. 

\subsection{Proofs Related to the Code Construction}

\paragraph{Properties of Construction~\ref{const:2lvl}}

For the given choice of codes, $\codeC(\codeA_{[2]}, \codeB_{[2]})$ as given in \Cref{const:2lvl} is $(\eseta,\esetb,w)$-PBEC according to \Cref{gcc:ecc}.
For $|\esetaa| = q^{\caa\cdot n +o(n)}$, $|\esetbb| = q^{\cbb\cdot n +o(n)}$, the inner codes $\codeB_1,\codeB_2$ on the GV bound have dimensions
\begin{align*}
    k_1 = n - n\cdot \caa - o(n),\\
    k_2 = n - n\cdot \cbb - o(n).
\end{align*}
The code $\codeA_1$ can be chosen MDS for $m \leq q^{k_1-k_2}$.
Due to
\[
q^{k_1-k_2} \geq q^{n\cdot\cbb} q^{-n\cdot\caa- o(n)} = |\esetbb| q^{o(n)} |\esetaa|^{-1},
\]
$|\esetbb|\cdot q^{-o(n)} \geq |\esetaa|\cdot m$ is sufficient.
Then, $\codeA_1$ has dimension $K_1 = \max\{m-2\wtw,0\}$.
The dimension of $\codeC(\codeA_{[2]}, \codeB_{[2]})$ is computed as $\max\{m-2w,0\}(k_1-k_2) + m\cdot k_2$.
Writing $\W = \wtw/n$, we obtain
\[
n m \cdot \left(\max\{1-2\W,0\} (\cbb-\caa) + (1-\cbb)\right)- m\cdot o(n), 
\]
and the statement on the rate follows by dividing by $n\cdot m$.

\paragraph{Properties of Construction~\ref{const:3lvl}}

For the given choice of codes, $\codeC(\codeA_{[3]}, \codeB_{[3]})$ as given in \Cref{const:3lvl} is $(\eseta,\esetb,w)$-PBEC according to \Cref{gcc:ecc}.
For $|\esetaa| = q^{\caa\cdot n +o(n)}$, $|\esetab| = q^{\cab\cdot n +o(n)}$, $|\esetbb| = q^{\cbb\cdot n +o(n)}$, the inner codes $\codeB_1,\codeB_2,\codeB_3$ on the GV bound have dimensions
\begin{align*}
    k_1 = n - n\cdot \caa - o(n),\\
    k_2 = n - n\cdot \cab - o(n),\\
    k_3 = n - n\cdot \cbb - o(n).\\
\end{align*}
The code $\codeA_1$ can be chosen MDS for $m \leq q^{k_1-k_2}$.
Due to
\[
q^{k_1-k_2} \geq q^{n\cdot\cab} q^{-n\cdot\caa- o(n)} = |\esetab| q^{o(n)} |\esetaa|^{-1},
\]
$|\esetab|\cdot q^{-o(n)} \geq |\esetaa|\cdot m$ is sufficient.
Similarly,  $\codeA_2$ can be MDS for $|\esetbb|\cdot q^{-o(n)} \geq |\esetab|\cdot m$.
Then, $\codeA_1$ and $\codeA_2$ have dimensions $K_1 = \max\{m-2\wtw,0\}$, $K_2 = m-\wtw$.
The dimension of $\codeC(\codeA_{[2]}, \codeB_{[2]})$ is computed as 
\[
\max\{m-2w,0\}(k_1-k_2) + (m-\wtw)(k_2-k_3) + m\cdot k_3.
\]
Writing $\W = \wtw/n$, we obtain
\begin{multline*}
n m \cdot (\max\{1-2\W,0\} (\cab-\caa) \\
+(1-\W)(\cbb-\caa) + (1-\cbb))- m\cdot o(n), 
\end{multline*}
and the statement on the rate follows by dividing by $n\cdot m$.
\paragraph{Proof of \Cref{comparison:2lvl_3lvl_gv}} The proof follows by considering each case individually.
\begin{align*}
    \caa\cbb&\geq\cab^2,& \!\!\!2\W&\leq 1: & \!\!\!\RGV -\RThree  &=  (\cab-\caa)\W ,\\
    \caa\cbb&\geq\cab^2,& \!\!\!2\W&\geq 1: & \!\!\!\RGV -\RThree  &=  (\cab-\caa) (1{-}\W),\\
  \!\!  \caa\cbb&\leq\cab^2,& \!\!\!2\W&\leq 1: & \!\!\!\RGV -\RThree  &=  (\cbb-\cab)\W ,\\
    \caa\cbb&\leq\cab^2,& \!\!\!2\W&\geq 1: & \!\!\!\RGV -\RThree  &=  (\cbb-\cab) (1{-}\W).
\end{align*}